%% file: Arvix_han.tex
\documentclass[12pt]{article}
\usepackage[top=3cm,bottom=3.5cm,left=3cm,right=3cm]{geometry}
\usepackage{graphicx}
\pdfoutput=1
\usepackage{mathrsfs}
\usepackage{amsfonts}
\usepackage{amsthm}
\usepackage{color}
\usepackage{amsmath}
\usepackage{setspace}
\usepackage{multirow}
\usepackage{cite}
\usepackage{pstricks}
\usepackage{fancybox}
\usepackage{multirow}
\usepackage{indentfirst}
\usepackage{graphicx}

\usepackage{tikz}
\usetikzlibrary{arrows,positioning}

\newcommand\pubdate{\today}

\def\support{\footnote{Corresponding author: lwu@pku.edu.cn}}

\def\Title#1{\begin{center} {\Large #1 } \end{center}}
\def\Author#1{\begin{center}{  #1} \end{center}}
\def\Address#1{\begin{center}{ \it #1} \end{center}}

\newenvironment{Abstract}{\begin{quotation}  }{\end{quotation}}

\input econfmacros.tex
\begin{document}
\begin{titlepage}
%\pubblock
\pubdate
\vfill
\Title{Equity Market Impact Modeling: an Empirical Analysis for Chinese Market}
\vfill
\Author{Shiyu Han$^{a}$, Lan Wu$^{a,}$\support, and Yuan Cheng$^{a}$\\}
\Address{
$^{a}${\small School of Mathematical Sciences, Peking University, Beijing 100871, China}}
%\Address{\napoli}
\vfill
\begin{Abstract}
Market impact has become a subject of increasing concern among academics and industry experts. We put forward a price impact model which considers the heteroscedasticity of price in the time dimension and dependency between permanent impact and temporary impact. We discuss and derive the extremum of the expectation of permanent impact and realized impact by constructing several special trading trajectories. Given our use of a large trade and quote tick records of 17,213,238,343 compiled from the Chinese stock market, the model assessment ultimately suggest that our model is better than Almgren's model. Interestingly, the result of random effect analysis indicates the parameter $\alpha$, which is the exponent of the impact function, is a constant with a value of around 0.7 across all stocks. Our model and empirical result would give academia some insight of mechanism of Chinese market, and can be applied to algorithm trading.
\end{Abstract}
\vfill
\end{titlepage}
\def\thefootnote{\fnsymbol{footnote}}
\section{Introduction}
A transaction cost consists of two parts, the direct cost and the indirect cost. The study of transaction cost can provide some information about the structure of market, and also some evidence about the efficiency of the market. The indirect cost is difficult to observe, measure, and control. Yet it can be optimized through trading techniques. Thus, the indirect cost is a subject of concern in academia and industry.

The former researchers have developed the basis of impact model. Basing on the model, there are a series of problems such as the optimization problem, parameters estimation and the impact functions. Some literatures also involve the application of the impact model.

\cite{kraus1972price} first decomposed price impact into permanent impact and temporary impact. They pointed out that the permanent impact was caused by a change in investors' expectations due to information gleaned from trading. The temporary impact, on the other hand, was caused by a short non-equilibrium in trading or the controlling of quotes by dealers, and it vanished after a period of time. \cite{kyle1985continuous} introduced the linear impact function, while \cite{torre1997barra} pointed that the impact function should be the power function. \cite{huberman2004price} proved that the permanent impact function should be linear, i.e., the power index must equal to 1 or it contradicted the non-arbitrage assumption. Moreover, the permanent impact function should be independent of the length of trading period. The temporary impact function, however, could be nonlinear and conformed more closely to microstructure theory and the empirical standard.

\cite{almgren2001optimal} obtained the optimal trading strategy given the conditions of the linear impact function. They also established the efficient frontier for the mean and variance of the optimal trade trajectories with different levels of risk aversion, and identified the closed-form solutions of the corresponding optimal trading strategies.\cite{almgren2003optimal} extended the linear temporary impact function to the power function. \cite{almgren2005direct} used the results attained by \cite{huberman2004price} to estimate the impact function based on data equity at Citigroup. They showed that permanent impact cost was irrelative to the trading duration and manner, while temporary impact cost was closely related to them. Moreover, the estimation of the power exponent for permanent impact was close to 0.8, while that of temporary impact was close to 0.6.

On the convexity of the impact function, \cite{madhavan1991bayesian} and \cite{hausman1992ordered} pointed out that the it was a concave function of the absolute value of the transaction volume. Unlike \cite{almgren2005direct}, who assumed that the permanent impact function was linear, they had assumed that the temporary impact function was strictly concave. \cite{monch2005optimal} found that the impact function was convex while researching the German security market. \cite{bhattacharya1991insiders} and \cite{spiegel2000asymmetric} showed that the impact function was a nonlinear S-form. \cite{chan1993institutional} and \cite{bikker2007market} studied the impact caused by different initiatives of buying and selling.

\cite{bikker2007market} considered how factors other than changes in price influenced trading. \cite{almgren2001optimal} attempted to calculate the liquidity-adjusted VaR and adapted it to the context of risk management. \cite{Lim2014Dynamic} considered dynamic portfolio management based on impact cost. \cite{draviam2011dynamic} converted the optimal execution trajectory problem with explicit and impact costs into an optimal stochastic control problem with a restricted condition.

Our motivation is to study the feature of price impact in Chinese stock market. We suppose that there exists a consistent impact function in a specific market, here we focus on Chinese market. We improve the previous impact models developed in \cite{almgren2001optimal} and introduce a statistical inference to assess the model. Specifically, we propose a more explicit form of the permanent impact function and temporary impact function and give the joint distribution of the statistics of the permanent impact and the temporary impact. Consequently, we acquire a strong explanation of Chinese market impact. Further, instead of looking for the optimal executing trajectory to minimize the utility function in \cite{almgren2001optimal}(the mean-variance utility), we derive the extremum of the expectation of permanent impact and realized impact by constructing corresponding trading trajectories. In empirical analysis, we calibrate the parameters in the impact models and made a model assessment on whole data in Chinese market. The empirical process and results should have a very strong practical meaning.

In Section 2, we present some preliminary work. In Section 3, we derive our main result. In Section 4, we describe our data. Section 5 is our empirical research. The paper ends up with a conclusion in Section 6.

\section{Preliminary}
In this section, we set up the process of the asset price and quantify the permanent impact and the temporary impact. For the sake of convenience, most of the variables and assumptions in our model follow from the framework in \cite{almgren2005direct}.
\subsection{Variables}
\subsubsection*{Volume time}
In the real market, the volatility of the asset price varies with the time over the course of one trading day. More specifically, an identically sized order can cause greater impact in the middle of the day than at the opening or closing of the trading day. To describe this phenomenon, we use volume time instead of physical time in our model. The mapping from real time to volume time can be defined as
\begin{equation}
\kappa(t)=t_{open}+V(t)\frac{t_{close}-t_{close}}{V},\quad t\in[t_{open},t_{close}],
\end{equation}
where V(t) is the transaction volume up to time t, V is the total transaction volume over the entire trading day, and $t_{open}$ and $t_{close}$ represent the opening and closing times of a trading day, respectively. In the following context, if time is not specified as physical time, it should be assumed the volume time.

\subsubsection*{Price variables}
As introduced in \cite{almgren2005direct}, the impact caused by executing an order contained a temporary compoenent, which was caused by a transient liquidity drought or other reasons and vanished after a period of time with the liquidity replenished. Here we define the beginning of the execution of an order as time 0, the end of the execution as time T, and the time at which the temporary impact vanishes as $T_{post}$, which is greater than T. Correspondingly, we define the prices at different times as
\begin{center}
 \begin{tabular}{c|r}
 \hline
 $S_0$ & the price before the impact occurs\\

  $S_{post}$ & the price when the temporary impact vanishes\\

  $\bar S$& the volume weighted average of the realized price\\
  \hline
 \end{tabular}
 \end{center}

$S_0$ represents the price before the impact occurs, $S_{post}$ represents the price when the temporary impact vanishes, and $\bar S$ indicates the volume weighted volume of the realized price in the execution period $[0,T]$.
\subsubsection*{Impact statistics}
  We decompose the impact into two components, the permanent impact and the temporary impact. The permanent impact reflects the information released during the trading process. The temporary impact reflects the loss to the counterparties caused by positive trading. The permanent impact is lasting, while the temporary impact vanishes as time passes. Two statistics to describe these two impacts are defined as follows:
\begin{align}
Permannent\;impact:\quad &I=\frac{S_{post}-S_0}{S_0},
\label{Idefine}\\
Realized\;impact:\quad &J=\frac{\bar S-S_0}{S_0}.
\label{Jdefine}
\end{align}
According to these two impact statistics, the temporary impact can be expressed as $I-J/2$.

\subsection{Price Process}
We assume asset price $S_t$ follows the arithmetic Brownian motion,
\begin{equation}
dS_t=S_0g(v_t)dt+S_0\sigma dB_t,\quad t\geq 0,
\label{S_process}
\end{equation}
where the trading velocity $v_t$ is deterministic, $g(v_t)$ in the drift term is the permanent impact function, which depends on trading velocity $v_t$ and $g(v_t)=0$ when $t>T$, and $B_t$ is the standard Brownian motion. It is important to point out that asset price $S_t$ in (\ref{S_process}) is the real price but not the transaction price, which we denote by $\tilde S_t$ in the following assumption:
\begin{equation}
\tilde S_t=S_t+S_0h(v_t),
\end{equation}
where $h(v)$ is the temporary impact function and depends on the trading velocity as well.

Under the process assumption above, we can calculate the concrete form of (\ref{Idefine}) and (\ref{Jdefine}) as follows (a detailed calculation is provided in the appendix):
\begin{align}
I&=\int_0^Tg(v_t)dt+\sigma B_{T_{post}}\label{Istat},\\
J&=\int_0^T\frac{T-t}{T}g(v_t)dt+\frac{1}{T}\int_0^Th(v_t)dt+\frac{\sigma}{T}\int_0^TB_tdt\label{Jstat}.
\end{align}

The impact functions $g(\cdot)$ and $h(\cdot)$ in (\ref{Istat}) and (\ref{Jstat}) have been studied by various researchers, as discussed in the introduction. In our work, we adopt the power functions as previous literatures:
\begin{align}
g(v)=\gamma sgn(v)|v|^\alpha,\\
h(v)=\eta sgn(v)|v|^\beta,
\end{align}

\section{Main Result}
\subsection{The joint distribution of $I$ and $J$}
 At first, we derive the joint distribution of $I$ and $J$ as the following proposition.

\newtheorem{theorem}{Propostion}
\begin{theorem}
\label{Propostion1}
$I$ and $J$ defined in (\ref{Istat}) and (\ref{Jstat}) follow Bivariate Normal Distribution as $(I,J)^T\sim N(\mu,\Sigma)$, where
\[\mu=\begin{pmatrix}\int_0^Tg(v_t)dt\\\int_0^T\frac{T-t}{T}g(v_t)dt+\frac{1}{T}\int_0^Th(v_t)dt\end{pmatrix},\quad
\Sigma=\begin{pmatrix}\sigma^2T_{post}&\frac{1}{2}\sigma^2T\\\frac{1}{2}\sigma^2T&\frac{1}{3}\sigma^2T\end{pmatrix}.\]
\end{theorem}

\begin{proof}
According to the property of Gaussian process, $I$ and $J$ follow Bivariate Normal Distribution. Basing on (\ref{Istat}) and (\ref{Jstat}), we can easily obtain the covariance matrix by the property of Ito integral.
\begin{align*}
Var(I)=E[I-E(I)]^2&=E(\sigma^2B_{T_{post}}^2)=\sigma^2T_{post},\\
Var(J)=E[J-E(J)]^2&=E[\frac{\sigma}{T}\int_0^T(T-t)dB_t]^2\\
                  &=\frac{\sigma^2}{T^2}\int_0^T(T-t)^2dt\\
                  &=\frac{1}{3}\sigma^2T.\\
\end{align*}
The covariance of $I$ and $J$ can be calculated as follows:
\begin{align*}
Cov(I,J)&=E[(I-E(I))(J-E(J))]\\
        &=E[\sigma \int_0^{T_{post}}dB_t\cdot\frac{\sigma}{T}\int_0^T(T-t)dB_t]\\
        &=\frac{\sigma^2}{T}E[\int_0^TdB_t\int_0^T(T-t)dB_t]\\
        &=\frac{\sigma^2}{T}\int_0^T(T-t)dt\\
        &=\frac{1}{2}\sigma^2T.
\end{align*}
Therefore, $I$ and $J$ follow the bivariate normal distribution as $(I,J)^T\sim N(\mu,\Sigma)$, where
\[\mu=\begin{pmatrix}\int_0^Tg(v_t)dt\\\int_0^T\frac{T-t}{T}g(v_t)dt+\frac{1}{T}\int_0^Th(v_t)dt\end{pmatrix},\quad
\Sigma=\begin{pmatrix}\sigma^2T_{post}&\frac{1}{2}\sigma^2T\\\frac{1}{2}\sigma^2T&\frac{1}{3}\sigma^2T\end{pmatrix}.\]
\end{proof}

\subsection{Optimal Impact}
Impact always exists in the trading process, though the magnitude varies. Instead of attempting to describe or even control for impact cost in the trading process, market participants may pay attention to the description of the impact directly. More specifically, in terms of expectations, market participants may ask: What is the optimum of the impact for a transaction of a fixed number of shares over a specified time interval? In what follows, we provide a specific solution of this problem given the explicit form of the impact function.

As mentioned above, we first suppose there are $X$ shares to be transacted over a given time interval $[0,T]$ (note that the time here is volume time) in a specific direction, i.e., buy or sell. Note that we want to find the optimum of the impact by transaction. In the following discussion, we denote the trajectory of trading by $x_t$, with an initial position of $x_0=0$ and a final position of $x_T=X$, which keeps a buy direction. Furthermore, the trading velocity $v_t$ does not change the sign. Similarly, we can analysis the opposite direction of sell.

Following the previous part of the paper, we take $I$ as the permanent impact, which exists even after the final transaction, while $J$ is the realized impact, which is the average impact during the trading process. Therefore, we can formulate the best trading trajectory to maximize or minimize the expected impact.

We propose the following optimizations:
\begin{align}
\begin{cases}
\min\,or\,\max\;E[I]\\
s.t.\quad x_0=0,\;x_T=X\;and\;v_t\geq0,
\end{cases}
\label{minmaxI}
\end{align}
and
\begin{align}
\begin{cases}
\min\,or\,\max\;E[J]\\
s.t.\quad x_0=0,\;x_T=X\;and\;v_t\geq0,
\end{cases}
\label{minmaxJ}
\end{align}
where $v_t=\dot{x_t}$.

The optimization problems (\ref{minmaxI}) and (\ref{minmaxJ}) can be generalized to the following variational problems:
\begin{align}
\min\;or\;\max\;\int_0^TF(t,x_t,v_t)dt.
\label{minmax}
\end{align}
In the permanent impact case, we have
\begin{equation}
F(t,x_t,v_t)=g(v_t),
\label{Fpermanent}
\end{equation}
in the realized impact case, we have
\begin{equation}
F(t,x_t,v_t)=\frac{T-t}{T}g(v_t)+\frac{1}{T}h(v_t).
\label{Ftemporary}
\end{equation}

To solve problem (\ref{minmax}), we consider the Euler-Lagrange equation:
\begin{equation}
F_x(t,x,v)=\frac{d}{dt}F_{v}(t,x,v).
\end{equation}
By substituting (\ref{Fpermanent}) and (\ref{Ftemporary}) into the Euler-Lagrange equation respectively, we obtain two differential equations. In (\ref{Fpermanent}) case, the equation is:
\begin{equation}
\label{ODE:I}
\gamma v\alpha(\alpha-1)v^{\alpha-2}=0,
\end{equation}
while in (\ref{Ftemporary}) case, the equation is:
\begin{equation}
\label{ODE:J}
-\frac{1}{T}\gamma\alpha v^{\alpha-1}+\frac{T-t}{T}\gamma\alpha(\alpha-1)v^{\alpha-2}\dot{v}
+\frac{1}{T}\eta\beta(\beta-1)v^{\beta-2}\dot{v}=0.
\end{equation}

The optimization problems (\ref{minmaxI}) and (\ref{minmaxJ}) can be solved through equation (\ref{ODE:I}) and equation (\ref{ODE:J}) with the boundary conditions.

Equations (\ref{ODE:I}) and (\ref{ODE:J}) are mainly parameterized by $\alpha$ and $\beta$, and they involve nonlinear problems for which no closed-form solutions can be obtained. Nevertheless, we can still obtain analytical results in special cases, which are shown in what follows.

\begin{theorem}
\label{Propostion2}
The solutions of optimization problem (\ref{minmaxI}) and (\ref{minmaxJ}) for some specific case are as bellow,
\begin{enumerate}
\item[(1)] For problem (\ref{minmaxI})\\
when $\alpha<1$, $\min E(I)=0$, $\max E(I)=\gamma T^{1-\alpha}X^\alpha$,\\
when $\alpha=1$, $\min E(I)=\max E(I)=\gamma X$,\\
when $\alpha>1$, $\min E(I)=\gamma T^{1-\alpha}X^\alpha$, $\max E(I)=+\infty.$
\item[(2)] For problem (\ref{minmaxJ})\\
when $\alpha=1,\beta<1,\beta\neq\frac{1}{2}$,\\
$~~~~~~\min E(J)=0$,\\
$~~~~~~\max E(J)=\frac{\eta^2\beta^2(\beta-1)}{\gamma T(2\beta-1)}[(\frac{\gamma}{\eta\beta}T+C_1)^{\frac{2\beta-1}{\beta-1}}
  -C_1^{\frac{2\beta-1}{\beta-1}}]+\gamma C_2$,\\
when $\alpha=1,\beta=\frac{1}{2}$,\\
$~~~~~~\min E(J)=0$,\\
$~~~~~~\max E(J)=\frac{\eta^2}{4\gamma T}\log(1+\frac{2\gamma T}{\eta C_1})$,\\
when $\alpha=1,\beta=1$,\\
$~~~~~~\min E(J)=\frac{\eta X}{T}$,\\
$~~~~~~\max E(J)=\frac{\eta X}{T}+\gamma X$,\\
when $\alpha=1,\beta>1$,\\
$~~~~~~\min E(J)=\frac{\eta^2\beta^2(\beta-1)}{\gamma T(2\beta-1)}[(\frac{\gamma}{\eta\beta}T+C_1)^{\frac{2\beta-1}{\beta-1}}
  -C_1^{\frac{2\beta-1}{\beta-1}}]+\gamma C_2$,\\
  $~~~~~~\max E(J)=+\infty$,\\
when $\beta=1,\alpha<1$,\\
$~~~~~~\min E(J)=\frac{\eta X}{T}$,\\
$~~~~~~\max E(J)=\gamma(\frac{\alpha-2}{\alpha-1})^{\alpha-1}X^\alpha T^{1-\alpha}+\frac{\eta X}{T}$,\\
when $\beta=1,1<\alpha\leq2$,\\
$~~~~~~\min E(J)=\frac{\eta X}{T}$,\\
$~~~~~~\max E(J)=+\infty$,\\
when $\beta=1,\alpha>2$,\\
$~~~~~~\min E(J)=\gamma(\frac{\alpha-2}{\alpha-1})^{\alpha-1}X^\alpha T^{1-\alpha}+\frac{\eta X}{T}$,\\
$~~~~~~\max E(J)=+\infty$.
\end{enumerate}
where $C_1$, $C_2$ satisfy the following equations,
  \begin{align*}
  \frac{\eta(\beta-1)}{\gamma}(\frac{\gamma}{\eta\beta}T+C_1)^{\frac{\beta}{\beta-1}}+C_2=X,\\
  \frac{\eta(\beta-1)}{\gamma}C_1^{\frac{\beta}{\beta-1}}+C_2=0.
  \end{align*}
\end{theorem}

\begin{proof}
In order to determine whether a solution generated by the Euler-Lagrange equation is a maximal or a minimal solution, we construct three series of trading trajectories.

Given a $m>0$, the first trajectory is
\begin{equation}
\label{eg1}
\begin{cases}
x_t=\frac{X}{T^m}t^m,\\
v_t=\frac{mX}{T^m}t^{m-1},
\end{cases}
\quad t\in[0,T],
\end{equation}
which is denoted as TA.

Given a $m>0$, the second trajectory is
\begin{equation}
\label{eg2}
\begin{cases}
x_t=X-\frac{X}{T^m}(T-t)^m,\\
v_t=\frac{mX}{T^m}(T-t)^{m-1},
\end{cases}
\quad t\in[0,T],
\end{equation}
which is denoted as TB.

And given a $m>0$, the third trajectory is
\begin{equation}
\label{eg3}
\begin{cases}
x_t=\frac{X}{\log\frac{T+m}{m}}\log\frac{T+m}{T-t+m},\\
v_t=\frac{X}{\log\frac{T+m}{m}}\frac{1}{T-t+m},
\end{cases}
\quad t\in[0,T],
\end{equation}
which is denoted as TC.

 Note that these three trading trajectories all satisfy the boundary condition so that $x_0=0$, $x_T=X$ and $v_t\geq0$.
\begin{description}
\item[(1)] For problem (\ref{minmaxI}) \mbox{}\par
\begin{description}
  \item[(a) $\alpha\neq1$] \mbox{}\par
  In this case, equation (\ref{ODE:I}) has solution $v=const>0$ and therefore $v=X/T$. The expected permanent impact is
  \begin{equation}
  E[I]=\gamma\int_0^T(\frac{X}{T})^\alpha dt=\gamma T^{1-\alpha}X^\alpha.
  \end{equation}

  Note that with the trading trajectory TA of (\ref{eg1}), we have
  \begin{align*}
  E[I]=\frac{\gamma X^\alpha T^{1-\alpha}m^\alpha}{\alpha m-\alpha}\xrightarrow{m\rightarrow\infty}
  \begin{cases}
  \infty \quad &\alpha>1,\\
  0      \quad &\alpha<1.
  \end{cases}
  \end{align*}

  Hence when $\alpha>1$, the solution of the differential equation (\ref{ODE:I}) gets a minimal expected impact, the lower bound of the permanent impact is $\gamma T^{1-\alpha}X^\alpha$, and the upper bound is $\infty$. Similarly, when $\alpha<1$, the lower bound is $0$ and the upper bound is $\gamma T^{1-\alpha}X^\alpha$.
\end{description}

  \begin{description}
  \item[(b) $\alpha=1$] \mbox{}\par
  In this case, equation (\ref{ODE:I}) degenerates, and the permanent impact function becomes $g(v)=\gamma v$. Accordingly, the expected permanent impact is
  \begin{equation}
  E[I]=\int_0^T\gamma vdt=\gamma X=const,
  \end{equation}
  which means that any trading trajectory has the same expected permanent impact, i.e., the upper bound and lower bound are both $\gamma X$.

  We remark that, with a total $X$ shares to be transacted, when $\alpha=1$, the permanent impact is a linear pattern and consequently the expected permanent impact is a constant regardless of trading velocity. When $\alpha<1$, a trading trajectory with uniform velocity leads to the maximal expected permanent impact. This means that one can obtain a lower expected permanent impact as long as the shares are transacted over a shorter time interval. When $\alpha>1$, similarly, a trading trajectory with uniform velocity leads to the minimal expected permanent impact, and a shorter time interval for the transaction leads to a greater expected permanent impact.
  \end{description}

\end{description}

\begin{description}
\item[(2)] For problem (\ref{minmaxJ}) \mbox{}\par
\begin{description}
  \item[(a) $\alpha=1$ and $\beta\neq 1$]\mbox{}\par
  In this case, equation (\ref{ODE:J}) reduces to
  \begin{equation}
  \label{ODE:Ja}
  \gamma=\eta\beta(\beta-1)v^{\beta-2}\dot{v}.
  \end{equation}
  The closed-form solution of (\ref{ODE:Ja}) is
  \begin{align}
  v_t&=(\frac{\gamma t}{\eta\beta}+C_1)^{\frac{1}{\beta-1}},\\
  x_t&=\frac{(\beta-1)\eta}{\gamma}(\frac{\gamma t}{\eta\beta}+C_1)^{\frac{\beta}{\beta-1}}+C_2,
  \end{align}
  where $C_1$ and $C_2$ satisfy the following equations:
  \begin{align}
  \frac{\eta(\beta-1)}{\gamma}(\frac{\gamma}{\eta\beta}T+C_1)^{\frac{\beta}{\beta-1}}+C_2=X,\\
  \frac{\eta(\beta-1)}{\gamma}C_1^{\frac{\beta}{\beta-1}}+C_2=0.
  \end{align}
  When $\beta\neq\frac{1}{2}$,
  \begin{equation}
  \label{BJA1}
  E[J]=\frac{\eta^2\beta^2(\beta-1)}{\gamma T(2\beta-1)}[(\frac{\gamma}{\eta\beta}T+C_1)^{\frac{2\beta-1}{\beta-1}}
  -C_1^{\frac{2\beta-1}{\beta-1}}]+\gamma C_2.
  \end{equation}
   When $\beta=\frac{1}{2}$,
   \begin{equation}
   \label{BJA2}
   E[J]=\frac{\eta^2}{4\gamma T}\log(1+\frac{2\gamma T}{\eta C_1}).
   \end{equation}
%  \begin{equation}
%  E[J]=\int_0^T\frac{T-t}{T}\gamma vdt+\frac{1}{T}\int_0^T\eta v^{\beta}dt=\frac{\gamma}{T}\int_0^TX_tdt+\frac{\eta}{T}\int_0^Tv^\beta dt
%  \end{equation}

  To verify the solution of differential equation (\ref{ODE:Ja}) is indeed corresponding a maximal or minimal realized impact, we choose the trading trajectory TA of (\ref{eg1}), the expected realized impact of which is
  \begin{equation}
   E[J]=\frac{\gamma}{m+1}X+\frac{\gamma X^\beta T^{-\beta}m^\beta}{\beta m-\beta+1}\xrightarrow{m\rightarrow\infty}
      \begin{cases}
      \infty, \quad&\beta>1,\\
      0,\quad &\beta<1.
      \end{cases}
  \end{equation}

  Hence, when $\beta<1$, we conclude that the solution of differential equation (\ref{ODE:Ja}) corresponds to the maximal expected realized impact. When $\beta>1$, we obtain the minimal expected realized impact. More specifically, when $\beta\neq 1/2$ and $\beta<1$, the lower bound of the expected realized impact is 0 and the upper bound is (\ref{BJA1}). When  $\beta=1/2$, the lower bound is $0$ and the upper bound is (\ref{BJA2}). On the other hand, when $\beta>1$, the lower bound is (\ref{BJA1}) and the upper bound is $\infty$.
  \end{description}

  \begin{description}
  \item[(b) $\alpha\neq1$ and $\beta=1$] \mbox{}\par
    In this case, equation (\ref{ODE:J}) reduces to
  \begin{equation}
  \label{ODE:Jb}
  \frac{1}{T}\gamma\alpha v^{\alpha-1}=\frac{T-t}{T}\gamma\alpha(\alpha-1)v^{\alpha-2}\dot{v}.
  \end{equation}
  When $\alpha>2$ or $\alpha<1$, the closed-form solution that corresponds to this case is
  \begin{align}
  v_t&=\frac{X}{T}(T-t)^{\frac{1}{1-\alpha}},\\
  x_t&=X[1-(\frac{T-t}{T})^{\frac{\alpha-2}{\alpha-1}}].
  \end{align}
  and when $1<\alpha\leq2$, there is no solution to the differential equation.

  Note that when $\beta=1$, the expected realized impact is
  \begin{align*}
  E[J]&=\int_0^T\frac{T-t}{T}\gamma v^\alpha dt+\frac{1}{T}\int_0^T\eta vdt\\
      &=\frac{\gamma}{T}\int_0^T(T-t)v^{\alpha}dt+\frac{\eta}{T}X\\
      &\geq\frac{\eta X}{T},
  \end{align*}
  and the solution of (\ref{ODE:Jb}) has expected realized impact
  \begin{equation}
  \label{BJB1}
  E[J]=\gamma(\frac{\alpha-2}{\alpha-1})^{\alpha-1}X^\alpha T^{1-\alpha}+\frac{\eta X}{T}.
  \end{equation}
  In order to verify whether (\ref{BJB1}) has the maximal or minimal expected realized impact, we also resort to trading trajectories TA of (\ref{eg1}) and TB of (\ref{eg2}). When $\alpha<2$, the expected realized impact of trading trajectory TA of (\ref{eg1}) is
  \begin{equation}
  E[J]=\frac{\gamma m^\alpha X^\alpha T^{1-\alpha}}{(\alpha m-\alpha+1)(\alpha m-\alpha+2)}+\frac{\eta X}{T}\xrightarrow{m\rightarrow\infty}\frac{\eta X}{T}.
  \end{equation}
  When $\alpha>1$, the expected realized impact of trading trajectory TB of (\ref{eg2}) is
  \begin{equation}
  E[J]=\frac{\gamma m^\alpha X^\alpha T^{1-\alpha}}{\alpha m-\alpha+2}+\frac{\eta X}{T}\xrightarrow{m\rightarrow\infty}\infty.
  \end{equation}
  This indicates that when $\alpha<1$, with a large enough $m$, trading trajectory TA of (\ref{eg1}) has a lower expected realized impact than the solution of differential equation (\ref{ODE:Jb}). Hence the solution results in the maximal expected realized impact,
  i.e., the lower bound is $\eta X/T$ and the upper bound is (\ref{BJB1}).

   Similarly, when $\alpha>2$, the solution of the differential equation results in the minimal expected realized impact, i.e., the lower bound is (\ref{BJB1}) while the upper bound is $\infty$. When $1<\alpha<2$, there are two sequences of trading trajectories that have a realized impact converging to $\eta X/T$
  and $\infty$, respectively. This implies that the realized impact interval is $(\eta X/T,\infty)$.

  When $\alpha=2$, the expected realized impact of trading trajectory TC of (\ref{eg3}) is
  \begin{align}
  E[J]&=\frac{\gamma X^2}{T\log^2\frac{T+m}{m}}\int_0^T\frac{T-t}{(T-t+m)^2}dt+\frac{\eta X}{T}\\
      &=\frac{\gamma X^2}{T}\frac{\log\frac{T+m}{m}-\frac{T}{T+m}}
      {\log^2\frac{T+m}{m}}+\frac{\eta X}{T}\xrightarrow{m\rightarrow0+}\frac{\eta X}{T}
  \end{align}
  Combined with trading trajectory TB of (\ref{eg2}), we can conclude that the realized impact interval is $(\eta X/T,\infty)$.
  \end{description}

  \begin{description}
    \item[(c) $\alpha=1$ and $\beta=1$] \mbox{}\par
  In this case, equation (\ref{ODE:J}) reduces to
  \begin{equation}
  \frac{\gamma}{T}=0.
  \end{equation}
  Hence this condition results in no solution to differential equation (\ref{ODE:J}). Indeed, note that the expected realized impact in this case is
  \begin{align*}
  E[J]&=\int_0^T\frac{T-t}{T}\gamma vdt+\frac{1}{T}\int_0^T\eta vdt\\
      &=\frac{\gamma}{T}\int_0^Tx_tdt+\frac{\eta X}{T}\in(\frac{\eta X}{T},\frac{\eta X}{T}+\gamma X).
  \end{align*}
  Similar to the former cases, trading trajectory TA of (\ref{eg1}) has an expected realized impact of
  \begin{equation}
  E[J]=\frac{\gamma}{m+1}X+\frac{\eta}{T}X\xrightarrow{m\rightarrow\infty}\frac{\eta X}{T},
  \end{equation}
  and trading trajectory TB of (\ref{eg2}) corresponds to
  \begin{equation}
  E[J]=\gamma X-\frac{\gamma X}{m+1}+\frac{\eta X}{T}\xrightarrow{m\rightarrow\infty}\frac{\eta X}{T}+\gamma X.
  \end{equation}
  This means that no trading trajectory can result in the maximal or minimal expected realized impact, i.e., the impact interval is $(\eta X/T,\eta X/T+\gamma X)$, which is consistent with the result that the the Euler-Lagrange equation has no solution.
  \end{description}

  \begin{description}
    \item[(d) Other cases] \mbox{}\par
   In addition to the above cases, the ODE (\ref{ODE:J}) is general nonlinear and its solvability remains more specific discussion.

   The summary remark of above are that, with a fixed $X$, when $\beta=1$, the temporary impact is a linear pattern. Similar to the permanent impact case, the temporary impact is invariant regardless of trading velocity. Then the realized impact is only affected by the permanent impact. When $\alpha<1$, a trading trajectory with a decreasing velocity gets a maximal permanent impact, which is positively related to $T$. When $\alpha=1$, the permanent impact is independent of $T$. When
   $1<\alpha\leq2$, the permanent impact can be any positive value regardless of $T$. When $\alpha>2$, a trading trajectory with increasing velocity gets a minimal permanent impact, which is negatively related to $T$.
   \end{description}

  \end{description}

\end{proof}

   The Propostion \ref{Propostion2} implies that we can control the permanent impact and the realized impact when $\alpha$ is not equal to 1. This means trader can manipulate the underlying price to some extent. By quickly buying a large volume of shares and then slowly selling part of the previously bought shares, she can open the position with nearly no price impact. By repeating this procedure, she can set up a large enough position with a small price impact. She buys slowly to push the price high enough, then she closes her position with a small impact in a way opposite to her position direction.

   Finally trader can profit just because the execute manner, which is a evidence of the inefficiency of the market. The farther that $\alpha$ deviates from 1, the more inefficient the market is and the greater the opportunity for this kind of profiting is.

\section{Data description}
Our empirical work is based on Chinese stock market data from January 2006 to January 2016. Initially we chose 2,542 stocks, giving us a corresponding total number of trade and quote (TAQ) tick records of 17,213,238,343. Taking into account that some stocks may have had invalid or incomplete data, we filtered the data sample. Ultimately, we selected 1,993 stocks that had valid data for at least five years, making the total number of tick records for the chosen stocks 16,189,238,181.

Note that we only had the tick-by-tick data; we did not have the true orders submitted by all traders. Therefore, we adopted the algorithm introduced by \cite{lee1991inferring} to identify the orders. The Lee-Ready algorithm suggested that we could identify the order as buyer-initiated if the transaction price was closer to the prevailing best ask price, and we could identify it as seller-initiated if the transaction price was closer to the prevailing best bid price. If the order was transacted at the prevailing mid-quote price, we identified it as buyer-initiated if the transacted price was higher than the prevailing transacted price, and seller-initiated if the transacted price was lower than the prevailing transacted price.

\section{Empirical analysis}
We apply the impact model to Chinese stock market to explore the mechanism of Chinese market. Specifically, we establish the statistical model of the price impact. Moreover, we assess our model by using some criteria including BIC and CRPS.

\subsection{Statistical model}
When estimating the parameters of the impact functions, \cite{almgren2005direct} did not consider the correlation between the temporary impact and the permanent impact. Instead, we calibrate the parameters based on the joint distribution of the impact statistics. Furthermore, \cite{almgren2005direct} only considered the heteroskedasticity across the underlying, while we consider the  heteroskedasticity of volatility in the time dimension besides that.

To estimate the parameters in our model, it is necessary to address the discrete version of our model based on our observed variables. We take 15 minutes as the time unit and analyse the impact caused by the executing of orders within each 15-minute segment. In the following of this part, the time tick $t$ denotes the number of the time unit of trading period. In addition, we choose $t_{post}$ as 15 minutes after the executing time interval.

In the discrete version of our model, we denote $(I_t,J_t)^T$ by $K_t$. And $v_t$, $V_t$, $\sigma_t$ are the trading velocity, trading volume, and volatility, respectively, within 15 minutes. In addition, $v_t$ is held to be positive if the trading direction is buying, and negative if it is selling. Note that the subindex $t$ of $\sigma_t$ above reflects the heteroskedasticity of volatility in the time dimension. We suppose that $\sigma_t$ explains the heteroskedasticity better than variables like turnover. The general way to calculate $\sigma_t$ is using the last price. But for the case of high frequency data, the return has negative correlation due to the bid-ask spread. As a result, the volatility is over-estimated. Instead we introduce a method for computation of $\sigma_t$ according to the so-called smart-price, which is a better estimate for the real price of the underlying in 15 minutes. The smart-price is defined as,
\[smartprice_t=(P_a*Q_b+P_b*Q_a)/(Q_a+Q_b),\]
where $P_a$, $P_b$, $Q_a$, $Q_b$ are the best ask price, the best bid price, and the corresponding volumes, respectively.

The distribution of $K_t$ follows Propostion \ref{Propostion1}. Recognizing that the trading volume across stocks varies, we can eliminate this effect by normalizing trading vocality $v_t$ by  corresponding trading volume $V_t$. Hence our model can be expressed as
\begin{align}
g(v_t)&=\gamma\sigma_t sgn(v_t)(\frac{v_t}{V_t})^\alpha,\\
h(v_t)&=\eta\sigma_tsgn(v_t)(\frac{v_t}{V_t})^\beta,
\end{align}
where $sgn(\cdot)$ is the sign function.

\subsection{Model estimation}
  With the joint distribution of impact statistics $K_t$ being Bivariate Normal as shown in Propostion \ref{Propostion1}, it is natural to use the maximum likelihood method to estimate the parameters of our model. The log-likelihood function of one stock over the all observations is
  \begin{align*}
  \mathcal{L}(\alpha,&\beta,\gamma,\eta,K_t,v_t,V_t,\sigma_t,t\in\mathcal{T})\\
  &=-\frac{1}{2}\sum_{t\in\mathcal{T}}\Big{[}(K_t-\mu_t)^T\Sigma_t^{-1}(K_t-\mu_t)+log(|\Sigma_t|)+2log(2\pi)\Big{]},
  \end{align*}
  where $\mathcal T$ is the data set, and
  \begin{align*}
  \mu_t&=\begin{pmatrix}T\gamma\sigma_tsgn(v_t)(\frac{v_t}{V_t})^\alpha\\
  \frac{1}{2}T\gamma\sigma_tsgn(v_t)(\frac{v_t}{V_t})^\alpha+\eta\sigma_tsgn(v_t)(\frac{v_t}{V_t})^\beta\end{pmatrix},\quad\\
\Sigma_t&=\begin{pmatrix}\sigma_t^2T_{post}&\frac{1}{2}\sigma_t^2T\\\frac{1}{2}\sigma_t^2T&\frac{1}{3}\sigma_t^2T\end{pmatrix}.
\end{align*}

First we estimate the parameters for each stock. The mean, maximal, minimal of time ticks for each stock are 28,564, 38,249 and 18,039 respectively. And the results are summarized in Figure \ref{figure1}. At first glance, the four parameters are distributed as bell-shaped curves. Figure \ref{figure1} shows that $\alpha$ is mainly distributed in the interval of $[0.6,0.75]$ and $\beta$ in $[0.6,0.8]$. In addition, we can see that the distributions of the parameters are concentrated, which motivates us to explore whether the parameters are constants shared by all stocks.

We proceed to estimate the consistent parameters jointly data over all stocks as well as in different boards, including the Main Board of the Shanghai Stock Exchange, the Main Board of the Shenzhen Stock Exchange, the Small and Medium Enterprise Board, and the Growth Enterprise Market. The total amount of time ticks for these five sets are 56,927,427, 25,442,267, 12,311,959, 13,730,481, and 5,442,720 respectively. Table \ref{table1} lists the results of the estimation, with the first column providing the abbreviations for the boards and the numbers in parentheses being the standard errors. We can see that $\alpha$ and $\beta$ are significantly greater than 0.5 yet less than 1. As for $\gamma$ and $\eta$, they are significantly different. These results attribute to the different dimensions of the impact functions $g(v)$ and $h(v)$.

The parameter $\alpha$ we obtained in the Chinese market is systematically less than that in the American market as obtained by \cite{almgren2005direct}, which indicates that the Chinese market is less efficient than the American market. This difference might be due to the immaturity of Chinese investors and the market.

\begin{table}[htbp]
 \footnotesize
 \centering
  \caption{The result of joint estimating for parameters.}
 \begin{tabular}{c|cccc}
 \hline
 &$\hat\alpha$&$\hat\beta$&$\hat\gamma$&$\hat\eta$\\
 \hline
 All&0.6866 (0.0012)&0.7090 (0.0024)&4.5713 (0.0088)&0.0520 (0.0002)\\
 SH&0.6799 (0.0023)&0.6986 (0.0045)&4.5796 (0.0164)&0.0547 (0.0004)\\
 SZ&0.6827 (0.0021)&0.6769 (0.0039)&4.5321 (0.0127)&0.0537 (0.0003)\\
 SME&0.6835 (0.0019)&0.7048 (0.0042)&4.4363 (0.0123)&0.0506 (0.0003)\\
 GME&0.6849 (0.0028)&0.7683 (0.0075)&4.4835 (0.0149)&0.0490 (0.0003)\\
 \hline
 \end{tabular}
  \label{table1}
 \end{table}
\subsection{Model assessment}
To justify our estimation results, we compare our model with the one presented in \cite{almgren2005direct}. We use the continuously ranked probability score (CRPS, \cite{matheson1976scoring}) and the BIC criterion.

CRPS is a proper scoring rule which was used by \cite{gneiting2007strictly} to assess the estimated probability distribution. When the estimated probability distribution coincides with the true probability distribution, the expectation for CRPS reaches its maximum, which is defined as
\[CRPS(F, y)=-\int_{-\infty}^{\infty}(F(x)-1_{\{y\leq x\}})^2dx,\]
where F is the estimated cumulative distribution and y is an observation. It can also be expressed as
\[CRPS(F,y)=\frac{1}{2}E_F|Y-Y'|-E_F|Y-y|,\]
where $Y$ and $Y'$ are independent variables with distribution F.

We compare the CRPS of our model with that of \cite{almgren2005direct}. According to its property, we expect to obtain a relatively greater CRPS. Thus, the comparision is assessed by the difference between the CRPS of our model and that of \cite{almgren2005direct} across all stocks. Figure \ref{figure2} (a-c) shows that the difference is mostly positive; therefore, our model outperforms the model in \cite{almgren2005direct}.

In addition, BIC is an often-used criterion for model selection. Note that BIC is consistent, and choosing the model with the minimum BIC is equivalent to choosing the model with the maximum post-probability. Thus we choose the model with the lower BIC. Similarly, we compare the BIC of our model and those of \cite{almgren2005direct} across all stocks. Figure \ref{fig24} shows that in most cases, our model has lower BIC that the model produced by \cite{almgren2005direct}.

Furthermore, we compare the CRPS and BIC for all stocks and the aforementioned boards. Figure \ref{table2} shows that the values of CRPS and BIC in our model are better than those of \cite{almgren2005direct}. Consequently, our model is more reliable.
%\begin{figure}[htbp]
%\renewcommand{\captionfont}{\footnotesize}
%\renewcommand{\captionlabelfont}{\footnotesize}
%\centering
%\subfloat[The difference of CRPS's for $I$]{
%\begin{minipage}{0.45\textwidth}
%\centering
%\includegraphics[width=\textwidth]{CRPS_I.eps}
%%\caption{The estimates of $\alpha$ across all stocks}
%\label{fig21}
%\end{minipage}
%}
%\hspace{10pt}
%\subfloat[The difference of CRPS's for $J$]{
%\begin{minipage}{0.45\textwidth}
%\centering
%\includegraphics[width=\textwidth,height=0.82\textwidth]{CRPS_J.eps}
%%\caption{The estimates of $\beta$ across all stocks}
%\label{fig22}
%\end{minipage}
%}
%\\
%\subfloat[The difference of CRPS's for $J-\frac{I}{2}$]{
%\begin{minipage}{0.45\textwidth}
%\centering
%\includegraphics[width=\textwidth]{CRPS_J_I_2.eps}
%%\caption{The estimates of $\gamma$ across all stocks}
%\label{fig23}
%\end{minipage}
%}
%\hspace{10pt}
%\subfloat[The difference for BIC's]{
%\begin{minipage}{0.45\textwidth}
%\centering
%\includegraphics[width=1.01\textwidth]{BIC.eps}
%%\caption{The estimates of $\eta$ across all stocks}
%\label{fig24}
%\end{minipage}
%}
%\caption{Histograms of the difference of the criteria value.}
%\label{figure2}
%\end{figure}

\begin{table}[htbp]
 \scriptsize
 \centering
  \caption{The result of joint estimating for parameters.}
 \begin{tabular}{c|cccc}
 \hline
 &CRPS for $I$&CRPS for $J$&CRPS for $J-\frac{I}{2}$&BIC\\
 &Our Alm2005&Our Alm2005&Our Alm2005&Our Alm2005\\
 \hline
 All&-0.0277 -0.0298&-0.0061 -0.0097&-0.0127 -0.0127&1867891.18 2040472.65\\
 SH&-0.0280 -0.0287&-0.0056 -0.0095&-0.0123 -0.0131&748337.13 1144537.64\\
 SZ&-0.0281 -0.0301&-0.0063 -0.0101&-0.0129 -0.0130&893682.96 955813.35\\
 SME&-0.0283 -0.0307&-0.0063 -0.0100&-0.0129 -0.0131&974733.32 1011478.44\\
 GME&-0.0281 -0.0292&-0.0066 -0.0103&-0.0128 -0.0134&514790.13 592355.86\\
 \hline
 \end{tabular}
  \label{table2}
 \end{table}

\subsection{Random effect of $\alpha$}
The results of the consistent estimate of $\alpha$ indicate that all of the stocks in the market may share the same $\alpha$. To verify this conjecture, we convert our model to a random effect model. More specifically, we assume that $\alpha$ is distributed as $N(\mu_\alpha,\sigma_\alpha^2)$. We introduce a hypothesis testing $H_0:\sigma_\alpha=0\,v.s.\,H_1:\sigma_\alpha>0$. If we reject $H_0$, we conclude $\alpha$ has random effect, which means that all of the stocks do not share the same $\alpha$. On the contrary, we believe $\alpha$ is a constant.

As in previous work, we analysed the random effect based on all of the stocks and the boards. The first column in Table \ref{table3} is $\hat\sigma_\alpha$ for our model, while the third column is for Almgren's model. The second and fourth columns are the p-values corresponding to the two models. In our model, we cannot reject the hypothesis $H_0$, while the result in Almgren's model significantly rejects $H_0$. Hence all of the stocks in the Chinese stock market share the same $\alpha$ in our model. As it is more consistent, our model performs better than Almgren's model in explaining the Chinese stock market.
\begin{table}[htbp]
 \centering
 \caption{Result of random effect analysis.}
 \begin{tabular}{c|cccc}
 \hline
  &$\hat\sigma_\alpha$&pvalue&$\hat\sigma_\alpha$&pvalue\\
 \hline
 ALL&0.0005&0.4776&0.0184&0.0007\\
 SH&0.0010&0.3123&0.0325&0.0012\\
 SZ&0.0009&0.3742&0.0278&0.0009\\
 SME&0.0009&0.1787&0.0225&0.0010\\
 GME&0.0012&0.3278&0.0055&0.0012\\
 \hline
 \end{tabular}
  \label{table3}
 \end{table}

\section{Conclusions}
Based on the assumptions of the price process and the impact functions, we formulated the models of permanent and temporary impact. We obtained the explicit optima of the expected realized impact and the permanent impact by means of constructing three series of trading trajectories for specific cases. Our findings suggest that when $\alpha$ is less than 1, the market is inefficient.

We improve upon the model in \cite{almgren2005direct} in two aspects. First, we take into account the heteroscedasticity in the time dimension by introducing the time-varying volatility model for a more in line with the reality of the market. In empirical analysis, we calibrate the market consistent parameters on the joint model with correlation between the permanent impact and the temporary impact. Moreover, CRPS and BIC indicate that our model is systematically better than Almgren's model. The test of random effect suggests $\alpha$ is a constant with statistically significant in the market. This result justifies that our model can reveal some mechanism of Chinese stock market.

We applied our model to the Chinese stock market with a large trade and quote dataset. The result of the empirical analysis implies that the Chinese market is less efficient than the American market. The findings provide some insights of the market and have significant guidance for regulators. For instance, regulators might enhance the trade charge for short swing trading or limit short swing trading. In practice, our model can also be applied to algorithm trading.

\newpage
\section*{Appendix A: Calculation of the Impact Variables}
\subsection*{Permanent Impact}
\begin{align*}
I=&\frac{S_{post}-S_0}{S_0}\\
 =&\frac{S_0+\int_0^TS_0g(v_t)dt+\int_0^{T_{post}}S_0\sigma dB_t-S_0}{S_0}\\
 =&\int_0^Tg(v_t)dt+\sigma B_{T_{post}}.
\end{align*}

\subsection*{Realized Impact}
\begin{align*}
J&=\frac{\bar S-S_0}{S_0}=\frac{\frac{1}{T}\int_0^T\tilde S_tdt-S_0}{S_0}\\
 &=\frac{\frac{1}{T}\int_0^T[S_t+S_0h(v_t)]dt-S_0}{S_0}\\
 &=\frac{1}{T}\int_0^T[(\int_0^tg(v_s)ds+\sigma B_t)+h(v_t)]dt\\
 &=\int_0^T\frac{T-t}{T}g(v_t)dt+\frac{1}{T}\int_0^Th(v_t)dt+\frac{\sigma}{T}\int_0^TB_tdt\\
 &=\int_0^T\frac{T-t}{T}g(v_t)dt+\frac{1}{T}\int_0^Th(v_t)dt+\sigma B_T-\frac{\sigma}{T}\int_0^TtB_t\\
 &=\int_0^T\frac{T-t}{T}g(v_t)dt+\frac{1}{T}\int_0^Th(v_t)dt+\frac{\sigma}{T}\int_0^T(T-t)dB_t.
\end{align*}

\bibliographystyle{plain}

\bibliography{Arxiv_han}

\end{document}

%% file: econfmacros.tex
\def\beq{\begin{equation}}
\def\eeq#1{\label{#1}\end{equation}}
\def\eeqn{\end{equation}}

\def\beqa{\begin{eqnarray}}
\def\eeqa#1{\label{#1}\end{eqnarray}}
\def\eeqan{\end{eqnarray}}

\let\bar=\overbar

\def\Dslash{\not{\hbox{\kern-4pt $D$}}}
\def\dslash{\not{\hbox{\kern-2pt $\del$}}}

\def\msb{{\bar{\ssstyle M \kern -1pt S}}}